\documentclass[12pt]{article}

\usepackage[affil-it]{authblk}
\usepackage{amsfonts}
\usepackage{amsmath,amsthm,amssymb,dsfont}
\usepackage{enumerate}
\usepackage[english]{babel}
\usepackage{graphicx}	
\usepackage[margin=2.9cm]{geometry}

\usepackage{todonotes}

\usepackage{hyperref}
\hypersetup{colorlinks=true,citecolor=blue,linkcolor=blue,filecolor=blue,urlcolor=blue,breaklinks=true}

\usepackage{tikz}
\usetikzlibrary{chains}
\usetikzlibrary{fit}
\usepackage{pgflibraryarrows}		
\usepackage{pgflibrarysnakes}		

\usepackage{epsfig}
\usetikzlibrary{shapes.symbols,patterns} 
\usepackage{pgfplots}

\theoremstyle{plain}
\newtheorem{theorem}{Theorem}[section]

\newtheorem{lemma}[theorem]{Lemma}

\newtheorem{corollary}[theorem]{Corollary}

\theoremstyle{definition}

\newtheorem{remark}[theorem]{Remark}

\newcommand*{\cI}{\mathcal{I}}

\newcommand*{\cN}{\mathcal{N}}
\newcommand*{\cM}{\mathcal{M}}
\newcommand*{\cP}{\mathcal{P}}

\newcommand*{\cR}{\mathcal{R}}
\newcommand*{\cS}{\mathrm{S}}

\newcommand*{\cU}{\mathcal{U}}

\newcommand*{\cX}{\mathcal{X}}

\newcommand*{\BB}{\mathrm{L}}
\newcommand*{\TC}{\mathrm{TC}}

\newcommand*{\ee}{\mathrm{e}}

\newcommand*{\N}{\mathbb{N}}
\newcommand*{\R}{\mathbb{R}}

\newcommand*{\eps}{\varepsilon}

\newcommand*{\D}{\mathrm{S}}
\newcommand*{\Q}{\mathrm{Q}}
\newcommand*{\Pos}{\mathrm{P}}
\newcommand*{\TPCP}{\mathrm{TPCP}}
\newcommand*{\supp}{\mathrm{supp}}
\newcommand*{\id}{\mathrm{id}}

\newcommand*{\tr}{\mathrm{tr}}
\newcommand*{\ket}[1]{| #1 \rangle}
\newcommand*{\bra}[1]{\langle #1 |}

\newcommand{\proj}[1]{|#1\rangle\!\langle #1|}

\newcommand{\norm}[1]{\left\lVert#1\right\rVert}

\newcommand*{\MD}{D_{\mathbb{M}}}

\newcommand*{\ci}{\mathrm{i}}  
\newcommand*{\di}{\mathrm{d}}  
\newcommand*{\dd}{|\hspace{-0.4mm}|}

\begin{document}

\title{Universal recovery maps and approximate sufficiency of quantum relative entropy}

\author[1]{Marius Junge}
\author[2]{Renato Renner}
\author[2]{David Sutter}
\author[3]{Mark M.~Wilde}
\author[4]{Andreas Winter}

 \affil[1]{Department of Mathematics, University of Illinois at Urbana-Champaign, Illinois 61801-2975, USA}
\affil[2]{Institute for Theoretical Physics, ETH Zurich, Switzerland}
\affil[3]{Hearne Institute for Theoretical Physics, Department of Physics and Astronomy, Center for Computation and Technology, Louisiana State University,  Baton Rouge, Louisiana 70803, USA}
\affil[4]{ICREA \&{} F\'{\i}sica Te\`{o}rica: Informaci\'{o} i Fen\`{o}mens
 Qu\`{a}ntics, Universitat Aut\`{o}noma de Barcelona, ES-08193
 Bellaterra (Barcelona), Spain.
}

\date{}

\maketitle

\begin{abstract}
The data processing inequality states that the quantum relative entropy between two states $\rho$ and $\sigma$ can never increase by applying the same quantum channel
$\mathcal{N}$ to both states. This inequality can be strengthened with a remainder term in the form of a distance between $\rho$ and the closest recovered state $(\mathcal{R} \circ \mathcal{N})(\rho)$, where $\mathcal{R}$ is a recovery map with the property that $\sigma = (\mathcal{R} \circ \mathcal{N})(\sigma)$. We show the existence of an \emph{explicit} recovery map that is \emph{universal} in the sense that it  depends only on $\sigma$ and the quantum channel $\mathcal{N}$ to be reversed. This result gives an alternate, information-theoretic characterization of the conditions for approximate quantum error correction.
\end{abstract}

\section{Introduction}
For two Hilbert spaces $A$ and $B$, let $\D(A)$ denote the set of density operators on $A$ and let $\TPCP(A,B)$ be the set of trace-preserving completely positive maps from $A$ to $B$.\footnote{A linear mapping $\cN_{A \to B}$ from $A$ to $B$ is said to be \emph{completely positive} if $(\cN_{A \to B} \otimes \cI_R)(\rho_{AR}) \geq 0$ for all $\rho_{AR} \geq 0$, where $R$ denotes an arbitrary reference system and $\cI_R$ denotes the identity map on $R$.  The mapping is additionally trace preserving if any positive semi-definite trace-class input operator is mapped to an output operator that has the same trace.} Let $\Q(A)$ denote some subset of $\D(A)$. A quantum channel $\cN \in \TPCP(A,B)$ is called \emph{sufficient} (or \emph{reversible}) \emph{with respect to} $\Q(A)$, if there exists a recovery map $\cR \in \TPCP(B,A)$ such that 
\begin{align}
(\cR \circ \cN)(\rho) = \rho \quad \textnormal{for all} \quad \rho \in \Q(A) \ .
\end{align}
Sufficiency of quantum channels has been studied extensively (see~\cite{Pet86,Petz88,mosonyi04,anna06} and references therein).

The quantum relative entropy between two states $\rho$ and $\sigma$ is defined as \cite{lindblad73}
\begin{align} \label{eq_defRelEnt}
D(\rho \| \sigma):= \sum_i \langle \phi_i \vert \rho (\log \rho - \log \sigma)\vert \phi_i \rangle =
\sum_{i,j} |\langle \phi_i \vert \psi_j\rangle|^2 [p(i) \log p(i) - p(i) \log q(j)]
\, ,
\end{align}
where $\rho = \sum_i p(i) \vert \phi_i \rangle \langle \phi_i \vert$
and 
$\sigma = \sum_j q(j) \vert \psi_j \rangle \langle \psi_j \vert$
are spectral decompositions of $\rho$
and $\sigma$
with $\{ \vert \phi_i \rangle \}_i$ and $\{ \vert \psi_j \rangle \}_j$ orthonormal bases. Note that if the support of $\rho$ is not contained in the support of $\sigma$, then $D(\rho \| \sigma) = +\infty$.
The \emph{data processing inequality} (also known as \emph{monotonicity of the relative entropy}) states that $D(\rho \| \sigma)$
is non-increasing under trace-preserving completely positive maps~\cite{lindblad75,uhlmann77}, i.e., 
\begin{align}
D(\rho \| \sigma) \geq D(\cN(\rho)\| \cN(\sigma)) \, , \label{eq_DPI}
\end{align}
where $\cN$ is a quantum channel. The data processing inequality is related to the sufficiency of $\cN$. 
As shown in~\cite{Pet86,Petz88,anna06,anna206}, a quantum channel $\cN \in \TPCP(A,B)$ is sufficient with respect to $\Q(A)$ if and only if $D(\rho \| \sigma) = D(\cN(\rho) \| \cN(\sigma))$ for all $\rho \in \Q(A)$ and $\sigma \in \Q(A)$. It is known that this is the case if and only if there exists a recovery map $\cR \in \TPCP(B,A)$ that simultaneously reverses the action of the physical evolution $\cN$ on both states~\cite{Pet86,Petz88,Pet03}, i.e., $(\cR \circ \cN)(\rho)=\rho$ and $(\cR \circ \cN)(\sigma)=\sigma$.

Let $\BB(B)$ denote the set of bounded operators on $B$, $\TC(A)$ the set of trace-class operators on $A$, and $\cN^{\dagger}$ the adjoint map of $\cN$.\footnote{Note that the adjoint $\cN^\dagger$ of $\cN$ is defined as the unique linear map satisfying $\left \langle a,\cN^\dagger(b) \right \rangle = \left \langle \cN(a),b \right \rangle$ for all $a \in \TC(A)$ and $b\in \BB(B)$, where $\left \langle a_1,a_2 \right \rangle:=\tr(a_1^\dagger a_2)$ is the Frobenius inner product.} 
We consider $\left \langle a_1, a_2 \right \rangle_{\omega}:=\tr(a_1^{\dagger} \omega^{\frac{1}{2}} a_2 \omega^{\frac{1}{2}})$, which is an inner product on the space of operators $\{a \in \BB(A)\,:\, \Pi_{\omega} a = a \Pi_{\omega} = a  \}$, where $\Pi_{\omega}$ is the projection onto the support of $\omega$.
If the Hilbert spaces $A$ and $B$ are assumed to be separable and $\sigma \in \TC(A)$ is positive semi-definite, then the
recovery map can be taken as the \emph{Petz recovery map} 
$\cP_{\sigma,\cN}$
(also known as the \emph{transpose map}), 
 defined as the adjoint of the solution to
\begin{align} \label{eq_EqPetzIntro}
 \langle a, \cN^{\dagger}(b)  \rangle_{\sigma} = \langle \cP^{\dagger}_{\sigma,\cN}(a),b  \rangle_{\cN(\sigma)} Ê\quad \textnormal{for all} \quad a \in \BB(A), b \in \BB(B) \ .
\end{align}
The Petz recovery map $\cP_{\sigma,\cN}$ is completely positive, trace non-increasing and unique on the support of $\cN(\sigma)$~\cite{Pet86,Petz88,Pet03,OP93}.
 In the case that the Hilbert spaces $A$ and $B$ are finite-dimensional, then, on the support of $\cN(\sigma)$, the Petz recovery map takes the form  
\begin{align} \label{eq_Petzmap}
\cP_{\sigma,\cN} \ : \  X_B \mapsto \sigma^{\frac{1}{2}} \cN^{\dagger}\bigl( \cN(\sigma)^{-\frac{1}{2}} X_B\, \cN(\sigma)^{-\frac{1}{2}} \bigr) \sigma^{\frac{1}{2}} \ .
\end{align}
(Following the standard convention, $\sigma^{-1}$ is defined to be the inverse of $\sigma$ on its support.)

The concept of sufficient statistics can be made robust.
For $\eps \in [0,1]$, a quantum channel $\cN \in \TPCP(A,B)$ is \emph{$\eps$-sufficient} with respect to $\Q(A)$ if there exists a recovery map $\cR_{\eps} \in \TPCP(B,A)$ such that~\cite{anna14} 
\begin{align}
\frac{1}{2}\norm{\rho - (\cR_\eps \circ \cN)(\rho)}_1 \leq \eps \quad \textnormal{for all} \quad \rho \in \Q(A) \, .
\end{align}
Together with the case $\eps=0$ discussed above and ideas from~\cite{WL12}, this motivates the question if there exists a stronger version of the data processing inequality. More precisely, one asks for the possibility of adding a non-negative term to the right-hand side of~\eqref{eq_DPI} that indicates how well $\rho$ can be recovered from $\cN(\rho)$. Such a relation would serve as an alternative characterization of approximate sufficient statistics.

An inequality that is closely related to the monotonicity of the relative entropy is the \emph{strong subadditivity} of quantum entropy~\cite{LieRus73_1,LieRus73}, which ensures that for any tripartite state $\rho_{ABC} \in \D(A\otimes B \otimes C)$ the conditional mutual information is non-negative, i.e., $I(A:C|B)_{\rho}:=H(AB)_{\rho}+H(BC)_{\rho}-H(ABC)_{\rho}-H(B)_{\rho} \geq 0$, where $H(A)_{\rho}:=-\tr(\rho_A \log \rho_A)$ denotes the von Neumann entropy.\footnote{The definition of conditional mutual information we have given is for finite-dimensional spaces and can be extended to the infinite-dimensional case \cite{S15}.} This inequality has been strengthened recently with a remainder term in the form of a distance to the closest recovered state. It was shown in~\cite{FR14}, that for any density operator $\rho_{ABC}$,  there exists $\cR_{B \to BC} \in \TPCP(B,B\otimes C)$ (the \emph{recovery map}) such that
\begin{align} \label{eq_FR}
I(A:C|B)_{\rho} \geq - 2 \log F(\rho_{ABC},\cR_{B \to BC}(\rho_{AB})) \ ,
\end{align}
where the \emph{fidelity} of $\rho$ and $\sigma$ is defined by~\cite{Uhl76} 
\begin{align} \label{eq_defFidelity}
F(\rho,\sigma):=\norm{\sqrt{\rho}\sqrt{\sigma}}_1 \, .
\end{align}
If $A$, $B$, and $C$ are finite-dimensional Hilbert spaces, on the support of $\rho_B$, $\cR_{B \to BC}$ can be taken as a rotated Petz recovery map, i.e., a trace-preserving completely positive map of the form
\begin{align} \label{eq_rotPETZ}
X_B \mapsto V_{BC} \rho_{BC}^{\frac{1}{2}} (\rho_B^{-\frac{1}{2}} U_B X_B U_B^{\dagger}\rho_B^{-\frac{1}{2}} \otimes \id_C) \rho_{BC}^{\frac{1}{2}} V_{BC}^{\dagger} \ ,
\end{align}
where $V_{BC}$ and $U_B$ are unitaries on $B\otimes C$ and $B$, respectively.

The result of~\cite{FR14}, whose proof is based on de Finetti type arguments and properties of R\'enyi entropies, has been extended and generalized in various ways. In~\cite{BHOS14}, based on the quantum state redistribution protocol \cite{DY08} and de Finetti type arguments, it was shown that the fidelity term can be replaced by a \emph{measured relative entropy} $\MD$, which is never smaller than the fidelity term, i.e.,
\begin{align} \label{eq_fernando}
I(A:C|B)_{\rho} \geq \MD(\rho_{ABC} \| \cR_{B \to BC}(\rho_{AB})) \geq- 2 \log F(\rho_{ABC},\cR_{B \to BC}(\rho_{AB})) \ .
\end{align}
 The measured relative entropy is defined as the supremum of the relative entropy with measured inputs over all positive operator-valued measures (POVMs) $\cM = \{M_x\}$, i.e., 
\begin{align} \label{eq_defMeasRelEnt}
\MD(\rho \| \sigma) := \sup \Bigl\{ D\bigl( \cM(\rho) \big\| \cM(\sigma) \bigr) : \cM(\rho) = \sum_{x} \tr(\rho M_x) \proj{x} \text{ with } \sum_{x} M_x = \id \Bigr\} \ ,
\end{align}
where $\{\ket{x}\}_x$ is a finite set of orthonormal vectors.
We note that the tighter bound from \cite{BHOS14} came at the cost of losing all information about the structure of the recovery map. In~\cite{SFR15}, it was shown that there exists a recovery map both satisfying~\eqref{eq_fernando} and possessing a universality property, in the sense that it only depends on the marginal $\rho_{BC}$. Furthermore, for a linearized version of~\eqref{eq_FR} it was shown that the recovery map has the form of a rotated Petz recovery map with commuting unitaries, i.e., a recovery map of the form in \eqref{eq_rotPETZ} where $V_{BC}$ and $U_B$ commute with $\rho_{BC}$ and $\rho_B$, respectively. 

In view of approximate sufficiency of quantum channels discussed above, it would be helpful to have a generalization of~\eqref{eq_FR} in terms of relative entropies. This has been established in~\cite{Wilde15} with a proof technique based on the notion of a R\'enyi generalization of a relative entropy difference \cite{SBW14} and Hadamard's three-line theorem. It was shown that for any two states $\rho$ and $\sigma$ on finite-dimensional Hilbert spaces with $\supp(\rho) \subseteq \supp(\sigma)$ and any channel $\cN$ there exists a recovery map $\cR$ such that $(\cR \circ \cN)(\sigma)=\sigma$ and
\begin{align} \label{eq_wilde}
D(\rho \| \sigma) - D\bigl(\cN(\rho) \| \cN(\sigma) \bigr) \geq - 2 \log F\bigl(\rho,(\cR\circ\cN)(\rho)\bigr) \ .
\end{align}
Furthermore, the recovery map was shown to be a rotated Petz recovery map with unitaries $U$ and $V$ in the algebra generated by
 $\sigma$ and $\cN(\sigma)$, respectively.
Very recently, another different proof technique was found~\cite{STH15}, based on the concavity and monotonicity of the operator logarithm, which shows that there exists a recovery map $\cR$ such that
\begin{align} \label{eq_STH}
D(\rho \| \sigma) - D\bigl(\cN(\rho) \| \cN(\sigma) \bigr)
& \geq \MD\bigl(\rho \| (\cR \circ \cN)(\rho)\bigr) \\
& \geq - 2 \log F\bigl(\rho,(\cR\circ\cN)(\rho)\bigr) \
\label{eq_STH_1}.
\end{align}
The recovery map was shown to be a convex combination of rotated Petz recovery maps with unitaries $U$ and $V$ in the algebra generated by $\sigma$ and $\cN(\sigma)$, respectively, and therefore satisfies $(\cR \circ \cN)(\sigma)=\sigma$.

Neither in~\cite{Wilde15} nor in~\cite{STH15} could the recovery map satisfying~\eqref{eq_wilde} and~\eqref{eq_STH}, respectively, be shown to be universal, in the sense that it could be taken independent of $\rho$. We note that by the Fuchs-van de Graaf inequality~\cite{fuchs99} the fidelity can be transferred into a trace distance term such that~\eqref{eq_wilde} and~\eqref{eq_STH_1} provide alternative characterizations for approximate sufficient statistics.

\paragraph{Result.} We show that for any non-negative operator $\sigma$ and for any channel $\cN$ there exists an \emph{explicit} and \emph{universal} recovery map $\cR_{\sigma,\cN}$  such that
\begin{align} \label{eq_res_intro}
D(\rho \| \sigma)  \geq D\bigl(\cN(\rho) \| \cN(\sigma) \bigr) - 2 \log F\bigl(\rho,(\cR_{\sigma,\cN}\circ\cN)(\rho)\bigr)
\end{align}
for all density operators $\rho$ such that $\supp(\rho) \subseteq \supp(\sigma)$. A consequence of the  universality of the recovery map $\cR_{\sigma,\cN}$ is that $(\cR_{\sigma,\cN}\circ\cN)(\sigma)=\sigma$. We note that $\rho$ and $\sigma$ are defined on separable Hilbert spaces.
 We refer to Theorem~\ref{thm_explicit} and Remark~\ref{rmk_mainSimple} for a more precise statement, and we note here that the result stated in Theorem~\ref{thm_explicit} is strictly stronger than~\eqref{eq_res_intro}.

\paragraph{History of the problem.} In 1973, the \emph{strong subadditivity} of quantum entropy \cite{LieRus73_1,LieRus73} was proven. It ensures that the conditional mutual information of any tripartite state  is non-negative. Two years later the \emph{data processing inequality}, or \emph{monotonicity of the quantum relative entropy under trace-preserving completely positive maps} was proven~\cite{lindblad75,uhlmann77}. This entropy inequality states that the quantum relative entropy cannot increase after applying a quantum channel to its arguments. Since then it has been realized that this fundamental theorem has numerous applications in quantum physics, and as a consequence, it was natural to ask if it would be possible to strengthen the result. This however turned out to be challenging. More recently, several conjectures regarding an improved data processing inequality have been put forward. (See~\cite{WL12} for one of these conjectures). In this paper, we prove~\eqref{eq_res_intro} with a recovery map that satisfies all the properties that have been conjectured to hold (see e.g.,~\cite{WL12}). 

\section{Main results}
Let $\Pos(A)$ denote the set of non-negative trace-class operators on a Hilbert space $A$.
For any $\sigma \in \Pos(A)$, we define 
\begin{align}
\D_{\sigma}(A):=\{\rho \in \D(A) : \supp(\rho) \subseteq \supp(\sigma)\} \ .
\end{align}
\begin{theorem}\label{thm_explicit}
Let $A$ and $B$ be separable Hilbert spaces. For any $\sigma \in \Pos(A)$, any $\rho \in \D_{\sigma}(A)$ and any $\cN \in \TPCP(A,B)$ we have
\begin{align} \label{eq_explcitMainRes}
D(\rho \dd \sigma) \geq D\bigl(\cN(\rho) \dd \cN(\sigma) \bigr) - 2 \int_{\R} \di t \, \beta_0(t) \,  \log F\bigl(\rho,( \cR^{\frac{t}{2}}_{\sigma,\cN} \circ\cN)(\rho)\bigr) \ ,
\end{align}
where the relative entropy and fidelity are defined in~\eqref{eq_defRelEnt} and~\eqref{eq_defFidelity}, respectively. The recovery map is given by 
\begin{align} \label{eq_PetzT}
\cR^t_{\sigma,\cN}  : X_B \mapsto \sigma^{-\ci t} \cP_{\sigma, \cN}\bigl( \cN(\sigma)^{\ci t} X_B\, \cN(\sigma)^{-\ci t} \bigr) \sigma^{\ci t}
\end{align}
and $\beta_0$ a probability density function on $\R$ defined by 
\begin{align} \label{eq_defBeta}
\beta_0(t):=\frac{\pi}{2}\bigl(\cosh(\pi t) + 1 \bigr)^{-1} \ .
\end{align}
The map $\cP_{\sigma, \cN}$ is the Petz recovery map, defined as the adjoint of the unique linear map
$\cP_{\sigma, \cN}^{\dag}$
satisfying~\eqref{eq_EqPetzIntro} with domain
$\BB(\supp(\cN(\sigma)))$ and range $\BB(\supp(\sigma))$.
If $A$ and $B$ are finite-dimensional, this unique linear map $\cP_{\sigma, \cN}$ is given by~\eqref{eq_Petzmap}. 
\end{theorem}

\begin{remark} \label{rmk_mainSimple}
Using the concavity of the logarithm and the fidelity~\cite[Exercise~9.20]{nielsenChuang_book}, Theorem~\ref{thm_explicit} can be simplified to
\begin{align} \label{eq_explcitMainResSimple}
D(\rho \dd \sigma) \geq D\bigl(\cN(\rho) \dd \cN(\sigma) \bigr) - 2  \log F\bigl(\rho,( \cR_{\sigma,\cN} \circ\cN)(\rho)\bigr) \ ,
\end{align}
where 
\begin{align}\label{eq_integral}
 \cR_{\sigma,\cN}(\cdot):= \int_{\R} \di t  \,\beta_0(t) \,  \cR^{\frac{t}{2}}_{\sigma,\cN}(\cdot)
\end{align}
on the support of $\cN(\sigma)$ with $\cR^t_{\sigma,\cN}$ and $\beta_0$ defined in Theorem~\ref{thm_explicit}.\footnote{The integral in~\eqref{eq_integral} can be understood as a Riemann sum with respect to weak convergence since we have a weakly continuous family of maps.}
\end{remark}

Figure~\ref{fig_beta} depicts the probability density $\beta_0$ as a function of $t\in \R$. 
We note that the recovery map $\cR_{\sigma,\cN}$ that satisfies~\eqref{eq_explcitMainResSimple} can be chosen such that it projects everything outside of the support of $\cN(\sigma)$ to zero.

\begin{figure}
\centering
\input{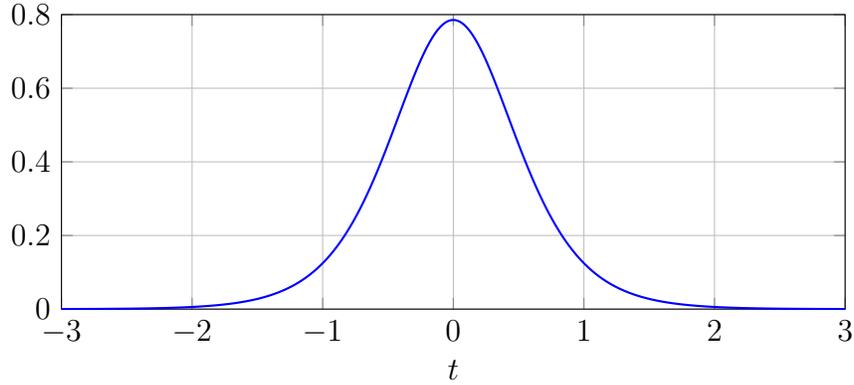}
\caption{This plot depicts the probability density $\beta_0$ defined in~\eqref{eq_defBeta} as a function of $t \inÊ\R$. We see that it is peaked around $t=0$ which corresponds to the Petz recovery map, i.e., $\mathcal{R}^{t=0}_{\sigma, \mathcal{N}} = \mathcal{P}_{\sigma, \mathcal{N}}$.}
\label{fig_beta}
\end{figure}

\begin{remark}
Inequality~\eqref{eq_explcitMainRes} together with the fact that the mapping $t\mapsto \cR_{\sigma,\cN}^t$ is continuous implies that for any $\sigma \in \Pos(A)$, $\rho \in \D_{\sigma}(A)$ and $\cN \in \TPCP(A,B)$ such that $D(\rho \| \sigma)=D(\cN(\rho) \| \cN(\sigma))$ we have $(\cR^t_{\sigma,\cN}\circ \cN)(\rho)=\rho$ and $(\cR^t_{\sigma,\cN}\circ \cN)(\sigma)=\sigma$ for  all $t \in \R$ with $\cR^t_{\sigma,\cN}(\cdot)$ defined in~\eqref{eq_PetzT}. This follows because $F(\omega,\tau) \in [0,1]$ and $F(\omega,\tau) = 1$ if and only if $\omega = \tau$ for density operators $\omega$ and $\tau$.
\end{remark}

\begin{remark}[Functoriality properties] \label{rmk_properties}
The recovery map $\cR_{\sigma,\cN}$  stated in~Remark~\ref{rmk_mainSimple} satisfies apart from~\eqref{eq_explcitMainResSimple} several desirable ``functoriality'' properties. Some of them have been stated in~\cite{WL12,LiWin14,Wilde15}. 
\begin{enumerate}
\item \textbf{Universality.} The recovery map does not depend on $\rho$. This follows directly from Remark~\ref{rmk_mainSimple}.

\item \textbf{Perfect reconstruction of $\sigma$ from $\cN(\sigma)$.} The recovery map satisfies $( \cR_{\sigma,\cN} \circ \cN)(\sigma)=\sigma$. This is clear from the fact that any rotated Petz map of the form in \eqref{eq_PetzT} perfectly recovers $\sigma$  \cite{Wilde15}, and thus so does any convex combination of these maps. Alternatively, as the recovery map predicted by Remark~\ref{rmk_mainSimple} that satisfies~\eqref{eq_explcitMainResSimple} is universal, the assertion follows by choosing $\rho = \sigma/\tr(\sigma)$.

\item \textbf{Normalization.} In case $\cN = \cI$, where $\cI$ denotes the identity map, we have $ \cR_{\sigma,\cN}(\cdot)=\Pi_{\sigma}(\cdot)\Pi_{\sigma}$, where $\Pi_{\sigma}$ denotes a projector onto the support of $\sigma$. Thus if $\sigma$ is faithful the recovery map is equal to the identity channel. 
This follows directly by~\cite[Section~4.1]{Wilde15} and by definition of the recovery map $\cR_{\sigma,\cN}(\cdot)$.

\item \textbf{Stabilization.} For any
$\sigma \in \Pos(A)$, any $\cN \in \TPCP(A,B)$, any reference system $R$, and any faithful $\tau \in \Pos(R)$, we have $ \cR_{\sigma\otimes \tau,\cN\otimes \cI_R}(\cdot) =  \cR_{\sigma,\cN} \otimes \cI_R(\cdot)$. This follows by combining~\cite[Section~4.2]{Wilde15} together with the normalization property discussed above.
\end{enumerate}
We note that by following~\cite[Sections~4.2 and 4.3]{Wilde15} it can be shown that the recovery map $\cR_{\sigma,\cN}$ fulfills some parallel and serial composition rules.
\end{remark}

The proof of Theorem~\ref{thm_explicit} consists of two parts. We first prove the statement for finite-dimensional Hilbert spaces $A$ and $B$ by employing a strengthened version of Hadamard's three-line theorem that is due to Hirschman \cite{H52}. By an approximation argument we show that the result remains valid for separable Hilbert spaces.


\section{Proof of Theorem~\ref{thm_explicit}}
\subsection*{\hypertarget{step_0_1}{Step 1:} Proof for  finite-dimensional Hilbert spaces}
In this step we assume that the Hilbert spaces $A$ and $B$ are finite-dimensional. 
Our proof of \eqref{eq_explcitMainRes} is similar to the approach taken in \cite{Wilde15}. There are two main ingredients: a R\'enyi generalization of a relative entropy difference \cite{SBW14} and Hirschman's improvement of the Hadamard three-line theorem \cite{H52}. We begin by recalling these two ingredients and then proceed to a proof of \eqref{eq_explcitMainRes}. 

For any $L \in \BB(A)$ the \emph{Schatten $p$-norm} is defined as
\begin{align}
\norm{L}_p := \bigl( \tr(|L|^p) \bigr)^{\frac{1}{p}} \quad \textnormal{for} \quad p\in [1,\infty) \ ,
\end{align}
where $|L|:=\sqrt{L^\dagger L}$. 
A R\'{e}nyi generalization of a relative entropy difference\footnote{The explanation in which sense this term is a relative entropy difference is given in~\eqref{eq:rel-ent-diff-a-1}.} is defined as
\cite{SBW14}%
\begin{equation}
\widetilde{\Delta}_{\alpha}(  \rho,\sigma,\mathcal{N})  :=
\frac{2\alpha}{\alpha-1}\log\left\Vert \left(  \left[  \mathcal{N}(
\rho)  \right]  ^{\frac{1-\alpha}{2\alpha}}\left[
\mathcal{N}(  \sigma)  \right]  ^{\frac{\alpha-1}{2\alpha}}\otimes \id_{E}\right)  U_{A\rightarrow BE}\, \sigma^{\frac{
1-\alpha}{2\alpha}}\rho^{\frac12}\right\Vert _{2\alpha}%
,\label{eq:renyi-diff}%
\end{equation}
where $\alpha\in(0,1)\cup(1,\infty)$, and $U_{A\rightarrow BE}$ is an
isometric extension of the channel $\mathcal{N}$. That is, $U_{A\rightarrow
BE}$ is a linear isometry satisfying $\operatorname{tr}_{E}(U_{A\rightarrow
BE}(  \cdot) U_{A\rightarrow BE}^{\dag})=\mathcal{N}\left(
\cdot\right)  $ and $U_{A\rightarrow BE}^{\dag}U_{A\rightarrow BE}=\id_{A}$. All
isometric extensions of a channel are related by an isometry acting on the
environment system~$E$, so that the definition in \eqref{eq:renyi-diff} is
invariant under any such choice. Recall also that the adjoint $\mathcal{N}%
^{\dag}$\ of a channel is given in terms of an isometric extension $U$ as
$\mathcal{N}^{\dag}(  \cdot)  =U^{\dag}(  (
\cdot)  \otimes \id_{E})  U$.
The following lemma was established in
\cite{SBW14}
for the case in which
$\rho$, $\sigma$, $\mathcal{N}(\rho)$,
and $\mathcal{N}(\sigma)$ are positive definite and was later extended in the appendix of \cite{Wilde15} to hold for the case in which $\rho  \in \D_{\sigma}(A)$:

\begin{lemma}
[\cite{SBW14,Wilde15}]
Let $A$ and $B$ be finite-dimensional Hilbert spaces.
The following limit holds for $\sigma \in \Pos(A)$, $\rho \in \D_{\sigma}(A)$, and  $\cN \in \TPCP(A,B)$:%
\begin{equation}
\lim_{\alpha\rightarrow1}\widetilde{\Delta}_{\alpha}(  \rho
,\sigma,\mathcal{N})  =D(\rho\Vert\sigma)-D\bigl(  \mathcal{N}(
\rho) \big\|  \mathcal{N}(  \sigma)  \bigr)
.\label{eq:rel-ent-diff-a-1}%
\end{equation}
\end{lemma}

For $\alpha=\frac12$, observe that%
\begin{align}
\widetilde{\Delta}_{\frac12}(  \rho,\sigma,\mathcal{N})    &
=-2\log\left\Vert \left(  \left[  \mathcal{N}(  \rho)  \right]
^{\frac12}\left[  \mathcal{N}(  \sigma)  \right]  ^{-\frac12}\otimes
\id_{E}\right)  U_{A\rightarrow BE}\,\sigma^{\frac12}\rho^{\frac12}\right\Vert _{1} \notag \\
&  =-2 \log  F\bigl(  \rho,\mathcal{P}_{\sigma,\mathcal{N}}\circ\mathcal{N}(  \rho)   \bigr)  \, ,
\end{align}
where $\cP_{\sigma,\cN}$ denotes the Petz recovery map defined in~\eqref{eq_Petzmap}.

The following lemma is based on Hirschman's improvement of the Hadamard three-line theorem \cite{H52}, and for completeness, we provide a proof in Appendix~\ref{sec:op-hirschman-proof}.

\begin{lemma}
\label{thm:op-hirschman}
Let $S:=\left\{  z\in\mathbb{C}:0\leq\operatorname{Re}\left\{  z\right\}
\leq1\right\} $ and let $G:S\rightarrow
L(\mathcal{H})$ be a bounded map that is holomorphic on the interior of $S$
and continuous on the boundary. Let $\theta\in(0,1)$ and define $p_{\theta}$
by%
\begin{equation}
\frac{1}{p_{\theta}}=\frac{1-\theta}{p_{0}}+\frac{\theta}{p_{1}}\ ,
\end{equation}
where $p_{0},p_{1}\in\left[  1,\infty\right]  $. Then the following bound
holds%
\begin{equation}
\log\bigl(  \left\Vert G(\theta)\right\Vert _{p_{\theta}}\bigr)  \leq
\int_{\R}\di t\ \Bigl(  \alpha_{\theta}(t)\log \bigl(  \left\Vert
G(\ci t)\right\Vert _{p_{0}}^{1-\theta}\bigr)  +\beta_{\theta}(t)\log\bigl(
\left\Vert G(1+\ci t)\right\Vert _{p_{1}}^{\theta}\bigr)  \Bigr) \ ,\label{eq:oper-hirschman}%
\end{equation}
where $\alpha_{\theta}(t)$ and $\beta_{\theta}(t)$ are defined by
\begin{align}
\alpha_{\theta}(t)   :=\frac{\sin(\pi\theta)}{2(1-\theta)\bigl(
\cosh(\pi t)-\cos(\pi\theta)\bigr)  } \qquad \textnormal{and} \qquad
\beta_{\theta}(t)  :=\frac{\sin(\pi\theta)}{2\theta\bigl( \cosh(\pi
t)+\cos(\pi\theta)\bigr)  }\ .\label{eq:pt}%
\end{align}
\end{lemma}

\begin{remark}
Fix $\theta\in(0,1)$. Observe that $\alpha_{\theta}(t),\beta_{\theta}(t)\geq0$
for all $t\in\mathbb{R}$ and we have
\begin{equation}
\int_{\R}\di t\ \alpha_{\theta}(t)=\int_{\R}%
\di t\ \beta_{\theta}(t)=1 \ ,
\end{equation}
(see, e.g., \cite[Exercise~1.3.8]{G08}) so that $\alpha_{\theta}(t)$ and
$\beta_{\theta}(t)$ can be interpreted as probability density functions.
Furthermore, the following limit holds
\begin{equation}
\lim_{\theta\searrow0}\beta_{\theta}(t)=\frac{\pi}{2\bigl(  \cosh(\pi
t)+1\bigr) } =\beta_{0}(t)\ ,\label{eq:dist-limit}%
\end{equation}
where $\beta_{0}$ is also a probability density function on $\R$.
\end{remark}%

We can now readily establish the desired result in \eqref{eq_explcitMainRes} for the finite-dimensional case.
In what follows, we abbreviate the isometric extension
$U_{A\rightarrow BE}$ of the channel $\mathcal{N}$
as $U$. Pick%
\begin{equation}
G(  z)  :=\left(  \left[  \mathcal{N}(  \rho)
\right]  ^{\frac{z}{2}}\left[  \mathcal{N}(  \sigma)  \right]
^{-\frac{z}{2}}\otimes \id_{E}\right)  U\, \sigma^{\frac{z}{2}}\rho^{\frac12}\ ,
\end{equation}
$p_{0}=2$, $p_{1}=1$, and $\theta\in\left(  0,1\right)  $, which fixes
$p_{\theta}=\frac{2}{1+\theta}$. The operator valued-function $G(
z)  $ satisfies the conditions needed to apply
Lemma~\ref{thm:op-hirschman}. For the choices above, we find that%
\begin{align}
\left\Vert G(  \theta)  \right\Vert _{\frac{2}{1+\theta}  }
&  =\left\Vert \left(  \left[  \mathcal{N}(  \rho)  \right]
^{\frac{\theta}{2}}\left[  \mathcal{N}(  \sigma)  \right]  ^{-\frac{\theta}{2}}\otimes \id_{E}\right)  U\, \sigma^{\frac{\theta}{2}}\rho^{\frac12}\right\Vert _{\frac{2}{
1+\theta}  }\ ,
\end{align}
and
\begin{align}
\left\Vert G(  \ci t)  \right\Vert _{2} &  =\left\Vert \left(  \left[
\mathcal{N}(  \rho)  \right]  ^{\frac{\ci t}{2}}\left[  \mathcal{N}(
\sigma)  \right]  ^{-\frac{\ci t}{2}}\otimes \id_{E}\right)  U\sigma^{\ci t}\rho
^{\frac12}\right\Vert _{2}\leq\left\Vert \rho^{\frac12}\right\Vert _{2}%
=1 \ ,\label{eq:M0} 
\end{align}
as well as
\begin{align}
\left\Vert G(  1+\ci t)  \right\Vert _{1} &  =\left\Vert \left(
\left[  \mathcal{N}(  \rho)  \right]^{\frac{1+\ci t}{2}}\left[  \mathcal{N}(  \sigma)  \right]  ^{-\frac{1+\ci t}{2}}\otimes \id_{E}\right)  U \,\sigma^{\frac{1+\ci t}{2}}\rho^{\frac12}%
\right\Vert _{1}\nonumber\\
& =\left\Vert \left(  \left[  \mathcal{N}(  \rho)  \right]
^{\frac{\ci t}{2}}\left[  \mathcal{N}(  \rho)  \right]^{\frac12}\left[
\mathcal{N}(  \sigma)  \right]  ^{-\frac{\ci t}{2}}\left[  \mathcal{N}(
\sigma)  \right]^{-\frac12}\otimes \id_{E}\right)  U\sigma^{\frac12}\sigma
^{\frac{\ci t}{2}}\rho^{\frac12}\right\Vert _{1}\nonumber\\
&  =\left\Vert \left(  \left[  \mathcal{N}(  \rho)  \right]
^{\frac12}\left[  \mathcal{N}(  \sigma)  \right]  ^{-\frac{\ci t}{2}}\left[
\mathcal{N}(  \sigma)  \right]  ^{-\frac12}\otimes \id_{E}\right)
U\sigma^{\frac12}\sigma^{\frac{\ci t}{2}}\rho^{\frac12}\right\Vert _{1}\nonumber\\
&  = F\bigl(  \rho,(\mathcal{R}_{\sigma,\mathcal{N}}^{\frac{t}{2}} \circ
\mathcal{N})(  \rho)    \bigr) \ .\label{eq:M1}%
\end{align}
Then we can apply the fact that $\left\Vert G(  \ci t)  \right\Vert
_{2}\leq1$ and \eqref{eq:oper-hirschman} to conclude that the following bound
holds for all $\theta\in(0,1)$
\begin{multline}
\log \left\Vert \left(  \left[  \mathcal{N(}\rho)\right]  ^{\frac{\theta}{2}}\left[  \mathcal{N(}\sigma)\right]^{-\frac{\theta}{2}}\otimes \id_{E}\right) U\, \sigma^{\frac{\theta}{2}}\rho^{\frac12}\right\Vert _{\frac{2}{1+\theta} }
\\
\leq \int_{\R}\di t\ \beta_{\theta}(t)\log \left( F\bigl(  \rho,(\mathcal{R}_{\sigma,\mathcal{N}}^{\frac{t}{2}} \circ
\mathcal{N})(  \rho)  \bigr)    ^{\theta} \right) \ ,
\end{multline}
which implies%
\begin{multline}
-\frac{2}{\theta}\log  \left\Vert \left(  \left[  \mathcal{N}(\rho)\right]  ^{\frac{\theta}{2}}\left[  \mathcal{N(}\sigma)\right]  ^{-\frac{\theta}{2}}\otimes \id_{E}\right)  U\sigma^{\frac{\theta}{2}}\rho^{\frac12}\right\Vert _{\frac{2}{1+\theta} }  \label{eq:critical-inequality}
\\
\geq-2\int_{\R}\di t\ \beta_{\theta}(t)\log  F\bigl(  \rho,(\mathcal{R}_{\sigma,\mathcal{N}}^{\frac{t}{2}} \circ
\mathcal{N})(  \rho)   \bigr) \  .
\end{multline}
Letting $\theta=\frac{ 1-\alpha}{\alpha}$, we see that this is the
same as%
\begin{equation}
\widetilde{\Delta}_{\alpha}(  \rho,\sigma,\mathcal{N})  \geq
- 2 \int_{\R}\di t\ \beta_{\frac{ 1-\alpha}{\alpha}}%
(t)\log F\bigl(  \rho,(\mathcal{R}_{\sigma,\mathcal{N}}^{\frac{t}{2}} \circ
\mathcal{N})(  \rho)    \bigr)
.\label{eq:alpha-bound}%
\end{equation}
Since the inequality in \eqref{eq:critical-inequality} holds for all
$\theta\in\left(  0,1\right)  $ and thus \eqref{eq:alpha-bound} holds for all
$\alpha\in\left(  \frac12,1\right)  $, we can take the limit as $\alpha\nearrow1$
and apply \eqref{eq:rel-ent-diff-a-1}, \eqref{eq:dist-limit}, and the dominated convergence theorem to conclude
that \eqref{eq_explcitMainRes} holds.

\begin{remark} \label{rmk_subnormalized}
If $A$ and $B$ are finite-dimensional Hilbert spaces the statement in Remark~\ref{rmk_mainSimple} can be slightly generalized. For any $\sigma \geq 0$ and any trace non-increasing completely positive map $\cN$ with Kraus operators $\{N_i\}$ such that $0 \neq \sum_i N_i^\dag N_i \leq \id$ the recovery map $ \cR_{\sigma,\cN}$ defined in Remark~\ref{rmk_mainSimple} satisfies \eqref{eq_explcitMainResSimple} for all subnormalized density operators $\rho \geq 0$ with $0 < \tr(\rho)\leq 1$ such that $\supp(\rho) \subseteq \supp(\sigma)$. This follows by the same argument given in Step~\hyperlink{step_1_1}{1} by using $U_{A\rightarrow BE} = \sum_i N_i \otimes \ket{i}_E$ as an extension of the trace non-increasing and completely positive map $\mathcal{N}$.
\end{remark}

\subsection*{\hypertarget{step_0_1}{Step 2:} Extension to infinite dimensions}

In this step the Hilbert spaces $A$ and $B$ are assumed to be separable (not necessarily finite-dimensional). We will show how to lift Theorem~\ref{thm_explicit} for finite-dimensional Hilbert spaces, such that it can apply to states $\rho$ and $\sigma$ and a channel $\cN$ associated with separable Hilbert spaces. This is accomplished via a limiting argument where we consider projected sequences that are finite-dimensional and therefore satisfy the desired inequality. In the limit we then obtain the statement for separable Hilbert spaces. This is a rather standard approach for generalizing statements proven for finite-dimensional Hilbert spaces to separable Hilbert spaces. 

Let $\{\Pi^a_A\}_{a \in \N}$ and $\{\Pi^b_B \}_{b \in \N}$ be sequences of finite-rank projectors on $A$ and $B$, respectively, that converge to $\id_A$ and $\id_B$, respectively, with respect to the weak operator topology, meaning that
\begin{equation}
\lim_{a\to \infty} \langle \psi \vert
\Pi^a_A \vert \phi \rangle = \langle \psi \vert
 \phi \rangle
\end{equation}
for all vectors 
$\vert \phi \rangle, \vert \psi \rangle \in A$ (similarly for $\Pi^b_B \to \id_B$). For $\sigma \in \Pos(A)$ and $\rho \in \D_{\sigma}(A)$ we consider projected versions
\begin{align} \label{eq_projRhoA}
\sigma^a :=\Pi_A^a \, \sigma \, \Pi_A^a \quad \textnormal{and} \quad \rho^a :=\Pi_A^a \, \rho \,  \Pi_A^a \ .
\end{align}
We note that the sequences $\{\rho^a \}_{a \in \N}$ and $\{\sigma^aÊ\}_{a \in \N}$ converge to $\rho$ and $\sigma$, respectively, in the trace norm (see, e.g.,~\cite{grumm} or Lemma~11.1 of \cite{holevo_book}).
Let $\cS^a$ be the set of non-negative operators that is generated by~\eqref{eq_projRhoA} for all $\rho \in \cS(A)$. For any $\cN \in \TPCP(A,B)$ we define its analogue with a projection at the input and output as
\begin{align}
\cN^{a,b}(\cdot) := \Pi_B^b \cN\bigl( \Pi_A^a ( \cdot) \Pi_A^a\bigr) \Pi_B^b \ .
\end{align}

Note that by combining Gr\"umm's theorem~\cite[Theorem~2.19]{simon_trace} with the boundedness of $\cN$ in the trace norm implies that $\cN^{a,b}$ converges to $\cN$ in the strong operator topology, i.e.,
\begin{equation}
\lim_{a,b\to \infty}
 \norm{\cN^{a,b}(\omega)-\cN(\omega)}_1 = 0
  \label{eq:channel-convergence-weak-op}
\end{equation}  
for all $\omega \in \TC(A)$.

We start by proving two lemmas that show how the difference of relative entropies and the fidelity, respectively, change when considering projected states.
\begin{lemma} \label{lem_one}
For any $\sigma \in \Pos(A)$, any $\rho \in \D_{\sigma}(A)$, and any $\cN \in \TPCP(A,B)$, we have
\begin{align}
\lim_{a \to \infty} D(\rho^a \| \sigma^a) = D(\rho \| \sigma)
\end{align}
and
\begin{align}
\liminf_{a \to \infty} \liminf_{b \to \infty}  D\bigl(\cN^{a,b}(\rho^a) \big\| \, \cN^{a,b}(\sigma^a) \bigr) \geq  D\bigl(\cN(\rho) \big\| \, \cN(\sigma) \bigr) \ .
\end{align}
\end{lemma}
\begin{proof}
Recall that Lindblad defined a relative entropy (which slightly differs from our definition of the relative entropy) for two positive trace class operators $\tau$ and $\omega$ as
\begin{align}
\bar D(\tau \| \omega) := \sum_{j} \bra{j} \left( \tau \log \tau - \tau \log \omega + \omega - \tau \right) \ket{j} = D(\tau \| \omega) + \tr(\omega) - \tr(\tau)\, ,
\end{align}
where $\{\ket{j}\}$ is a complete orthonormal set of eigenvectors of $\tau$ or $\omega$.
By Lemma~3 in~\cite{lindblad74} we have
\begin{align}
D(\rho^a \| \sigma^a)
&=\bar D(\rho^a \| \sigma^a) + \tr(\rho^a) - \tr(\sigma^a) \\
&\leq \bar D(\rho \| \sigma) + \tr(\rho^a) - \tr(\sigma^a)\\
&= D(\rho \| \sigma) + \tr(\rho^a - \rho) + \tr(\sigma - \sigma^a) \, .
\end{align}
Since the relative entropy is lower semicontinuous~\cite[Theorem~11.6]{holevo_book}, we obtain
\begin{equation}
D(\rho \| \sigma) \leq \liminf_{a \to \infty} D(\rho^a \| \sigma^a) \ .
\end{equation}
Combining these shows that
\begin{align} \label{eq_s1}
\lim_{a \to \infty} D(\rho^a \| \sigma^a) = D(\rho \| \sigma) \ .
\end{align}
The lower semicontinuity of the relative entropy implies that
\begin{align}
 \liminf_{a \to \infty} \liminf_{b \to \infty} D\bigl(\cN^{a,b}(\rho^a) \| \cN^{a,b}(\sigma^a) \bigr) 
& \geq \liminf_{a \to \infty} D\bigl(\cN(\rho^a) \| \cN(\sigma^a) \bigr) \\
& \geq D\bigl(\cN(\rho) \| \cN(\sigma) \bigr) \ .
\end{align}
This proves the assertion.
\end{proof}

\begin{lemma} \label{lem_PetzWellDefined}
The Petz recovery map $\cP_{\sigma^a,\cN^{a,b}}$ defined in~\eqref{eq_Petzmap} satisfies $\lim_{a\to \infty} \lim_{b \to \infty} \cP_{\sigma^a,\cN^{a,b}} = \cP_{\sigma,\cN}$, where $\cP_{\sigma,\cN}$ is a recovery map that is defined as the (unique) linear map with domain
$\operatorname{supp}(\cN(\sigma))$ and range
$\operatorname{supp}(\sigma)$
satisfying
\begin{align} \label{eq_EqPetz}
\left \langle a_2, \cN^{\dagger}(a_1) \right \rangle_{\sigma} = \left \langle \cP^{\dagger}_{\sigma,\cN}(a_2),a_1 \right \rangle_{\cN(\sigma)} Ê\quad \textnormal{for all} \quad a_1 \in \BB(B), \, a_2 \in \BB(A)\ ,
\end{align}
with the weighted inner product $\left \langle a, b \right \rangle_{\omega}:=\tr(a^{\dagger} \omega^{\frac{1}{2}} b \omega^{\frac{1}{2}})$ and $\omega$ positive semi-definite and trace class. 
\end{lemma}
\begin{proof}
We begin by outlining the proof.
As shown by Petz~\cite{Pet86,Petz88,Pet03} (see also \cite[Chapter~8]{OP93}), the map $\cP_{\sigma^a,\cN^{a,b}}$ defined in~\eqref{eq_Petzmap} is the unique linear map
with domain
$\operatorname{supp}(\cN^{a,b}(\sigma^a))$ and range
$\operatorname{supp}(\sigma^a)$
 satisfying
\begin{align} \label{eq_setEq}
\left \langle a_2, (\cN^{a,b})^{\dagger}(a_1) \right \rangle_{\sigma^a} = \left \langle \cP^{\dagger}_{\sigma^a,\cN^{a,b}}(a_2),a_1 \right \rangle_{\cN^{a,b}(\sigma^a)} Ê\quad \textnormal{for all} \quad a_1 \in \BB(B), \, a_2 \in \BB(A) \ ,
\end{align}
and the map $\cP_{\sigma,\cN}$ is the unique linear map
with domain
$\operatorname{supp}(\cN(\sigma))$ and range
$\operatorname{supp}(\sigma)$
 satisfying
\begin{align} \label{eq_setEq-inf-dim}
\left \langle a_2, (\cN)^{\dagger}(a_1) \right \rangle_{\sigma} = \left \langle \cP^{\dagger}_{\sigma,\cN}(a_2),a_1 \right \rangle_{\cN(\sigma)} Ê\quad \textnormal{for all} \quad a_1 \in \BB(B), \, a_2 \in \BB(A) \ .
\end{align}
We will first show that
\begin{align} \label{eq_s1_converge}
\lim_{a \to \infty} \lim_{b \to \infty} \left \langle a_2, (\cN^{a,b})^{\dagger}(a_1) \right \rangle_{\sigma^a}  
=\left \langle a_2, \cN^{\dagger}(a_1) \right \rangle_{\sigma} ,
\end{align}
for all $a_1 \in \BB(B),  a_2 \in \BB(A)$,
which by \eqref{eq_setEq} and \eqref{eq_setEq-inf-dim} implies that
\begin{align}
\lim_{a \to \infty} \lim_{b \to \infty} \left \langle \cP^{\dagger}_{\sigma^a,\cN^{a,b}}(a_2),a_1 \right \rangle_{\cN^{a,b}(\sigma^a)} = \left \langle \cP^{\dagger}_{\sigma,\cN}(a_2),a_1 \right \rangle_{\cN(\sigma)}
\label{eq:petz-converge-almost-final-1}
\end{align}
for all $a_1 \in \BB(B),  a_2 \in \BB(A)$.
After showing that
\begin{align}
\lim_{a \to \infty} \lim_{b \to \infty} \left \langle \cP^{\dagger}_{\sigma^a,\cN^{a,b}}(a_2),a_1 \right \rangle_{\cN^{a,b}(\sigma^a)} = 
\lim_{a \to \infty} \lim_{b \to \infty} \left \langle \cP^{\dagger}_{\sigma^a,\cN^{a,b}}(a_2),a_1 \right \rangle_{\cN(\sigma)} ,\label{eq:petz-converge-almost-final}
\end{align}
we can conclude from \eqref{eq:petz-converge-almost-final-1} and \eqref{eq:petz-converge-almost-final} that
\begin{align} 
\lim_{a \to \infty} \lim_{b \to \infty} \left \langle \cP^{\dagger}_{\sigma^a,\cN^{a,b}}(a_2),a_1 \right \rangle_{\cN(\sigma)} =
\left \langle \cP^{\dagger}_{\sigma,\cN}(a_2),a_1 \right \rangle_{\cN(\sigma)}. \label{eq:petz-converge-almost-final-2}
\end{align}
From there, we argue that this implies the convergence
$\lim_{a\to \infty} \lim_{b \to \infty} \cP_{\sigma^a,\cN^{a,b}} = \cP_{\sigma,\cN}$.

Thus, we need to establish \eqref{eq_s1_converge} and
\eqref{eq:petz-converge-almost-final}, and we begin by proving
\eqref{eq_s1_converge}. 
H\"older's inequality implies that
\begin{align}
& \left \vert \langle a_2, (\cN^{a,b})^\dagger(a_1) \rangle_{\sigma^a} - \langle a_2, \cN^\dagger(a_1) \rangle_{\sigma} \right \vert\notag \\
& = \left \vert \tr \, (\cN^{a,b})^\dagger(a_1)(\sigma^a)^{\frac{1}{2}} a_2^\dagger(\sigma^a)^{\frac{1}{2}}- \tr \, \cN^\dagger(a_1) \sigma^{\frac{1}{2}} a_2^\dagger \sigma^{\frac{1}{2}}  \right \vert \\
&=\left \vert \tr \, a_1 \big( \cN^{a,b}((\sigma^a)^{\frac{1}{2}} a_2^\dagger(\sigma^a)^{\frac{1}{2}}) - \cN(\sigma^{\frac{1}{2}} a_2^\dagger \sigma^{\frac{1}{2}}) \big)  \right \vert \\
&\leq \norm{a_1}_{\infty} \norm{\cN^{a,b}((\sigma^a)^{\frac{1}{2}} a_2^\dagger(\sigma^a)^{\frac{1}{2}})-  \cN(\sigma^{\frac{1}{2}} a_2^\dagger \sigma^{\frac{1}{2}}) }_1 \\
& \leq  \norm{a_1}_{\infty} \norm{\Pi_B^b  \cN((\sigma^a)^{\frac{1}{2}} a_2^\dagger(\sigma^a)^{\frac{1}{2}}) \Pi_B^b -  \cN((\sigma^a)^{\frac{1}{2}} a_2^\dagger(\sigma^a)^{\frac{1}{2}}) }_1 \nonumber \\
&\hspace{5mm} +  \norm{a_1}_{\infty} \norm{ \cN((\sigma^a)^{\frac{1}{2}} a_2^\dagger(\sigma^a)^{\frac{1}{2}})  -  \cN(\sigma^{\frac{1}{2}} a_2^\dagger \sigma^{\frac{1}{2}}) }_1 \, ,
\end{align}
where the final step uses the triangle inequality and that $\Pi_A^a (\sigma^a)^\frac{1}{2} = (\sigma^a)^\frac{1}{2}$. Gr\"umm's theorem~\cite[Theorem~2.19]{simon_trace} implies that for any fixed $a \in \N$ we have
\begin{align}
\lim_{bÊ\to \infty}  \norm{\Pi_B^b  \cN((\sigma^a)^{\frac{1}{2}} a_2^\dagger(\sigma^a)^{\frac{1}{2}}) \Pi_B^b -  \cN((\sigma^a)^{\frac{1}{2}} a_2^\dagger(\sigma^a)^{\frac{1}{2}}) }_1 = 0 \, ,
\end{align}
and hence it remains to show that
\begin{align} \label{eq_dddaviiiid}
\lim_{a \to \infty}  \norm{ \cN((\sigma^a)^{\frac{1}{2}} a_2^\dagger(\sigma^a)^{\frac{1}{2}})  -  \cN(\sigma^{\frac{1}{2}} a_2^\dagger \sigma^{\frac{1}{2}}) }_1 =0 \, .
\end{align}
To see this we first note that 
\begin{align}
 \norm{ \cN((\sigma^a)^{\frac{1}{2}} a_2^\dagger(\sigma^a)^{\frac{1}{2}})  -  \cN(\sigma^{\frac{1}{2}} a_2^\dagger \sigma^{\frac{1}{2}}) }_1 
 \leq \norm{\cN}_{1 \to 1} \norm{(\sigma^a)^{\frac{1}{2}} a_2^\dagger(\sigma^a)^{\frac{1}{2}} -\sigma^{\frac{1}{2}} a_2^\dagger \sigma^{\frac{1}{2}}}_1
\end{align}
and by the triangle inequality
\begin{align}
&\!\!\!\! \norm{(\sigma^a)^{\frac{1}{2}} a_2^\dagger(\sigma^a)^{\frac{1}{2}} -\sigma^{\frac{1}{2}} a_2^\dagger \sigma^{\frac{1}{2}}}_1\notag \\
& \leq \norm{(\sigma^a)^{\frac{1}{2}} a_2^\dagger(\sigma^a)^{\frac{1}{2}} - \sigma^{\frac{1}{2}} a_2^\dagger(\sigma^a)^{\frac{1}{2}}}_1 +\norm{\sigma^{\frac{1}{2}} a_2^\dagger(\sigma^a)^{\frac{1}{2}} - \sigma^{\frac{1}{2}} a_2^\dagger \sigma^{\frac{1}{2}}}_1 \\
&\leq \norm{(\sigma^a)^{\frac{1}{2}}- \sigma^{\frac{1}{2}}}_2 \norm{a_2}_{\infty} \norm{(\sigma^a)^\frac{1}{2}}_2 + \norm{\sigma^\frac{1}{2}}_2  \norm{a_2}_{\infty} \norm{(\sigma^a)^\frac{1}{2}-\sigma^{\frac{1}{2}}}_2\\
&\leq 2 \norm{a_2}_{\infty}  \norm{(\sigma^a)^\frac{1}{2}-\sigma^{\frac{1}{2}}}_2 \tr(\sigma) \\
&\leq 2 \norm{a_2}_{\infty} \norm{\sigma^a - \sigma}^\frac{1}{2}_1 \tr(\sigma) \, , \label{eq_esssDoone}
\end{align}
where second step follows by H\"older's inequality. The final step uses the Powers-St{\o}rmer inequality~\cite{Powers1970}.
Inequality~\eqref{eq_esssDoone} implies~\eqref{eq_dddaviiiid} and thus proves~\eqref{eq_s1_converge}.

We now prove \eqref{eq:petz-converge-almost-final}, by a reasoning that is
very similar to the above. 
H\"older's inequality implies that
\begin{align}
& \left \vert \langle  \cP^\dagger_{\sigma^a,\cN^{a,b}}(a_2), a_1 \rangle_{\cN^{a,b}(\sigma^a)} - \langle \cP^\dagger_{\sigma^a,\cN^{a,b}}(a_2), a_1  \rangle_{\cN(\sigma)}   \right \vert \nonumber \\
&\hspace{10mm}= \left \vert \tr\, a_1 \left( \cN^{a,b}(\sigma^a)^{\frac{1}{2}} \cP^\dagger_{\sigma^a,\cN^{a,b}}(a_2^\dagger)  \cN^{a,b}(\sigma^a)^{\frac{1}{2}}  - \cN(\sigma)^{\frac{1}{2}} \cP^\dagger_{\sigma^a,\cN^{a,b}}(a_2^\dagger)\cN(\sigma)^{\frac{1}{2}} \right)   \right \vert  \\
&\hspace{10mm} \leq  \norm{a_1}_{\infty} \norm{\cN^{a,b}(\sigma^a)^{\frac{1}{2}} \cP^\dagger_{\sigma^a,\cN^{a,b}}(a_2^\dagger)  \cN^{a,b}(\sigma^a)^{\frac{1}{2}}  -  \cN(\sigma)^{\frac{1}{2}} \cP^\dagger_{\sigma^a,\cN^{a,b}}(a_2^\dagger)\cN(\sigma)^{\frac{1}{2}}}_1 \, .
\end{align}
By the triangle inequality we find
\begin{align}
& \norm{\cN^{a,b}(\sigma^a)^{\frac{1}{2}} \cP^\dagger_{\sigma^a,\cN^{a,b}}(a_2^\dagger)  \cN^{a,b}(\sigma^a)^{\frac{1}{2}}  -  \cN(\sigma)^{\frac{1}{2}} \cP^\dagger_{\sigma^a,\cN^{a,b}}(a_2^\dagger)\cN(\sigma)^{\frac{1}{2}}}_1 \nonumber \\
&\hspace{15mm} \leq  \norm{\cN^{a,b}(\sigma^a)^{\frac{1}{2}} \cP^\dagger_{\sigma^a,\cN^{a,b}}(a_2^\dagger)  \cN^{a,b}(\sigma^a)^{\frac{1}{2}} - \cN(\sigma)^{\frac{1}{2}} \cP^\dagger_{\sigma^a,\cN^{a,b}}(a_2^\dagger)  \cN^{a,b}(\sigma^a)^{\frac{1}{2}}}_1 \nonumber \\
&\hspace{20mm}+\norm{\cN(\sigma)^{\frac{1}{2}} \cP^\dagger_{\sigma^a,\cN^{a,b}}(a_2^\dagger)  \cN^{a,b}(\sigma^a)^{\frac{1}{2}} - \cN(\sigma)^{\frac{1}{2}} \cP^\dagger_{\sigma^a,\cN^{a,b}}(a_2^\dagger)  \cN(\sigma)^{\frac{1}{2}}}_1 \\
&\hspace{15mm} \leq 2 \norm{\cP^\dagger_{\sigma^a,\cN^{a,b}}(a_2^\dagger)}_\infty  \norm{\cN^{a,b}(\sigma^a)^{\frac{1}{2}} - \cN(\sigma)^{\frac{1}{2}} }_2  \tr(\sigma)\\
&\hspace{15mm} \leq 2 \norm{a_2^\dagger}_\infty  \norm{\cN^{a,b}(\sigma^a) - \cN(\sigma) }^{\frac{1}{2}}_1 \tr(\sigma) \, ,
\end{align}
where the final step uses the Powers-St{\o}rmer inequality~\cite{Powers1970}. The triangle inequality implies that
\begin{align}
& \!\!\!\!\lim_{a \to \infty} \lim_{b \to \infty} \norm{\cN^{a,b}(\sigma^a) - \cN(\sigma) }_1 \notag \\
&\leq \lim_{a \to \infty} \lim_{b \to \infty}  \norm{\cN^{a,b}(\sigma^a) - \cN(\sigma^a) }_1 + \norm{ \cN(\sigma^a) - \cN(\sigma)}_1  \\
& \leq \lim_{a \to \infty} \norm{\sigma^a - \sigma}_1 \\
& = 0 \, ,
\end{align}
where the penultimate step uses Gr\"umm's theorem~\cite[Theorem~2.19]{simon_trace}. 

It remains to argue for the convergence $\lim_{a\rightarrow\infty}%
\lim_{b\rightarrow\infty}\mathcal{P}_{\sigma^{a},\mathcal{N}^{a,b}%
}=\mathcal{P}_{\sigma,\mathcal{N}}$. Consider that
\eqref{eq:petz-converge-almost-final-2} implies that%
\begin{multline}
\lim_{a\rightarrow\infty}\lim_{b\rightarrow\infty}\langle\phi^{\prime
}|\mathcal{N}(\sigma)^{\frac{1}{2}}\mathcal{P}_{\sigma^{a},\mathcal{N}^{a,b}%
}^{\dag}(|\phi\rangle\langle\psi|)\mathcal{N}(\sigma)^{\frac{1}{2}}%
|\psi^{\prime}\rangle
\\
=\langle\phi^{\prime}|\mathcal{N}(\sigma)^{\frac{1}{2}%
}\mathcal{P}_{\sigma,\mathcal{N}}^{\dag}(|\phi\rangle\langle\psi
|)\mathcal{N}(\sigma)^{\frac{1}{2}}|\psi^{\prime}\rangle \, ,
\end{multline}
for all $|\phi\rangle,|\psi\rangle\in A$ and $|\phi^{\prime}\rangle
,|\psi^{\prime}\rangle\in B$. This establishes convergence of $\mathcal{P}%
_{\sigma^{a},\mathcal{N}^{a,b}}^{\dag}$ to $\mathcal{P}_{\sigma,\mathcal{N}%
}^{\dag}$ for all $|\phi\rangle,|\psi\rangle\in A$ and $|\phi^{\prime}%
\rangle,|\psi^{\prime}\rangle\in\operatorname{supp}(\mathcal{N}(\sigma))$,
which in turn allows us to conclude convergence of $\mathcal{P}_{\sigma
^{a},\mathcal{N}^{a,b}}$ to $\mathcal{P}_{\sigma,\mathcal{N}}$.
\end{proof}

Before stating the following lemma, we introduce the shorthand
\begin{equation}
\cU_{\omega,t} : X \mapsto \omega^{\ci t} X \omega^{-\ci t} \ ,
\end{equation}
where $\omega$ is a positive semi-definite operator.

\begin{lemma} \label{lem_two}
For any $\sigma \in \Pos(A)$, any $\rho \in \D_{\sigma}(A)$, any $\cN \in \TPCP(A,B)$, and all $t\in \mathbb{R}$, the following limit holds
\begin{align}
\lim_{a \to \infty} \lim_{b \to \infty} F\bigl(\rho^a, (\cR^t_{\sigma^a,\cN^{a,b}} \circ \cN^{a,b})(\rho^a) \bigr) = F\bigl(\rho, (\cR^t_{\sigma,\cN} \circ \cN)(\rho) \bigr)\, .
\end{align}
\end{lemma}
\begin{proof}
We start with a reminder of a standard result: let $\omega$ and $\{\omega_n\}_{n \in \N}$ be positive semi-definite trace-class operators such that $\lim_{n \to \infty} \norm{\omega_n - \omega}_1=0$. Then, for any $t \in \R$ and any bounded linear operator $X$ the term $\omega_n^{\ci t} X \omega_n^{-\ci t}$ converges to $\omega^{\ci t} X \omega^{-Ê\ci t}$ for $n \to \infty$ in the weak operator topology. To see this we note that by Cauchy-Schwarz we have
\begin{align}
& \left| \langle \varphi, (\omega_n^{\ci t} X \omega_n^{-\ci t}- \omega^{\ci t} X \omega^{-\ci t}) \phi  \rangle \right|
\notag \\
& = \left| \langle \varphi, (\omega_n^{\ci t} X \omega_n^{-\ci t}- \omega_n^{\ci t} X \omega^{-\ci t}) + (\omega_n^{\ci t} X \omega^{-\ci t} - \omega^{\ci t} X \omega^{-\ci t})\phi  \rangle \right| \\
&\leq \left| \langle X \omega_n^{-\ci t} \varphi, (\omega_n^{-\ci t} - \omega^{-\ci t}) \phi  \rangle \right| + \left| \langle (\omega_n^{\ci t} - \omega^{\ci t}) \varphi, X \omega^{-\ci t} \phi  \rangle \right| \\
&\leq \norm{X} \norm{\varphi} \norm{(\omega_n^{-\ci t} - \omega^{-\ci t})\phi} + \norm{X} \norm{\phi} \norm{(\omega_n^{\ci t} - \omega^{\ci t})\varphi} \, .
\end{align}
We first note that for any fixed $k \in \N$ 
\begin{align}
\lim_{n \to \infty} \norm{(\Pi_k \omega_n^{-\ci t} \Pi_k- \Pi_k \omega^{-\ci t} \Pi_k) \phi} = 0 \, . \label{eq_DS1DS}
\end{align}
To see this we recall Duhamel's formula which shows that
\begin{align}
& \!\!\!\!(\Pi_k \omega_n^{-\ci t}\Pi_k  - \Pi_k \omega^{-\ci t}\Pi_k ) \notag \\
&=(\ee^{-\ci t \log \Pi_k \omega_n\Pi_k  }- \ee^{-\ci t \log \Pi_k \omega\Pi_k })  \\
&=  \ci \int_0^{t} \di s \,  \ee^{- \ci s \log \Pi_k \omega_n\Pi_k } (\log \Pi_k \omega\Pi_k  - \log \Pi_k \omega_n\Pi_k ) \ee^{- \ci (t-s) \log \Pi_k \omega \Pi_k } \, .
\end{align}
Theorem~X.3.7 from~\cite{bhatia_book} ensures that for any fixed $k \in \N$
\begin{multline}
\norm{\log \Pi_k \omega\Pi_k  - \log \Pi_k \omega_n\Pi_k} \\
\leq \norm{(\Pi_k \omega_n\Pi_k)^{-1}} \norm{\Pi_k \omega\Pi_k  - \Pi_k \omega_n\Pi_k} + O(\norm{\Pi_k \omega\Pi_k  - \Pi_k \omega_n\Pi_k}^2) \, ,
\end{multline}
where $\norm{( \Pi_k \omega_n\Pi_k)^{-1}} < \infty$ and $\lim_{n \to \infty}\norm{\Pi_k \omega\Pi_k  - \Pi_k \omega_n\Pi_k} =0$. This justifies~\eqref{eq_DS1DS}. By the triangle inequality we have
\begin{align}
& \!\!\!\! \norm{(\omega_n^{- \ci t} - \omega^{-\ci t}) \phi } \notag \\
&= \norm{(\omega_n^{- \ci t} - \omega^{-\ci t})  (\phi - \Pi_k \phi + \Pi_k \phi)}\\
& \leq \norm{(\Pi_k - \Pi_k + \id)(\omega_n^{-\ci t} \Pi_k- \omega^{-\ci t} \Pi_k) \phi} + \norm{(\omega_n^{- \ci t} - \omega^{-\ci t})  (\phi - \Pi_k \phi)}\\
&\leq  \norm{\Pi_k \omega_n^{-\ci t} \Pi_k- \Pi_k \omega^{-\ci t} \Pi_k) \phi} +\norm{(\id - \Pi_k)(\omega_n^{-\ci t} \Pi_k- \omega^{-\ci t} \Pi_k) \phi} \nonumber \\
&\hspace{20mm} + \norm{(\omega_n^{- \ci t} - \omega^{-\ci t})  (\phi - \Pi_k \phi)}\, ,
\end{align}
where for any fixed $n \in \N$
\begin{align}
\lim_{k \to \infty} \norm{(\id - \Pi_k)(\omega_n^{-\ci t} \Pi_k- \omega^{-\ci t} \Pi_k) \phi} = 0 
\end{align}
and
\begin{align}
\lim_{k \to \infty} \norm{(\omega_n^{- \ci t} - \omega^{-\ci t})  (\phi - \Pi_k \phi)} = 0 
\end{align}
We thus conclude that 
\begin{align}
\lim_{n \to \infty} \norm{(\omega_n^{- \ci t} - \omega^{-\ci t}) \phi }
&\leq \lim_{n \to \infty} \lim_{k \to \infty} \norm{(\Pi_k \omega_n^{-\ci t} \Pi_k- \Pi_k \omega^{-\ci t} \Pi_k) \phi}\\
&=\lim_{k \to \infty} \lim_{n \to \infty} \norm{(\Pi_k \omega_n^{-\ci t} \Pi_k- \Pi_k \omega^{-\ci t} \Pi_k) \phi}\\
&=0 \, ,
\end{align}
where in the penultimate step we swap the order of the limits~\cite[Theorem~7.11]{Rudin-76}.
The final step uses~\eqref{eq_DS1DS}.

We thus have the following convergences in the weak operator topology:
\begin{align}
\lim_{a \to \infty} \lim_{b \to \infty}
\cU_{\cN^{a,b}(\sigma^a),t}  = \cU_{\cN(\sigma),t}  \qquad \textnormal{and} \qquad
\lim_{a \to \infty} 
\cU_{\sigma^a,t}  = \cU_{\sigma,t} \, .
\end{align}
The serial concatenation of weakly converging channels converges to the serial concatenation of the limit, so that $\lim_{a \to \infty} \lim_{b \to \infty} \cR^t_{\sigma^a,\cN^{a,b}} \circ \cN^{a,b} =  \cR^t_{\sigma,\cN} \circ \cN
$, where we used the above and Lemma~\ref{lem_PetzWellDefined}. Then we can conclude that the fidelity converges because  it is continuous in its inputs (see, e.g.,~\cite[Lemma~B.9]{FR14}), and for our case considered here, convergence in the weak sense implies convergence in the trace norm \cite[Chapter~11]{holevo_book}. This proves the assertion.
\end{proof}

By invoking Theorem~\ref{thm_explicit} for finite-dimensional Hilbert spaces\footnote{Recall that by Remark~\ref{rmk_subnormalized}, Theorem~\ref{thm_explicit} remains valid for subnormalized states and a trace non-increasing completely positive map.} together with Lemmas~\ref{lem_one} and~\ref{lem_two},
we find for any $\rho \in \D_{\sigma}(A)$
\begin{align}
D(\rho \| \sigma) & = \lim_{a \to \infty} D(\rho^a \| \sigma^a) \\
&\geq \limsup_{a \to \infty} \limsup_{b \to \infty}   D\bigl(\cN^{a,b}(\rho^a) \| \cN^{a,b}(\sigma^a) \bigr) \notag \\
& \qquad - 2 \int_{\R}  \di t  \beta_0(t) 
\log F\bigl(\rho^a , (\cR^{\frac{t}{2}}_{\sigma^a,\cN^{a,b}} \circ \cN^{a,b})(\rho^a) \bigr) \\
&\geq D\bigl(\cN(\rho) \| \cN(\sigma) \bigr)  - 2 
\int_{\R} \ \di t \ \beta_0(t)  \ 
\log F\bigl(\rho, (\cR_{\sigma,\cN}^{\frac{t}{2}} \circ \cN)(\rho) \bigr) \ ,
\end{align}
where we also used Fatou's lemma that ensures 
\begin{align}
&\liminf_{a \to \infty} \liminf_{b \to \infty}  \int_{\R} \di t \, \beta_0(t) \Big(-2 \log F\bigl(\rho^a , (\cR^{\frac{t}{2}}_{\sigma^a,\cN^{a,b}} \circ \cN^{a,b})(\rho^a) \bigr)\Big) \nonumber \\
&\hspace{30mm}\geq \int_{\R} \di t  \, \beta_0(t) \liminf_{a \to \infty} \liminf_{b \to \infty}\Big( -2 \log F\bigl(\rho^a , (\cR^{\frac{t}{2}}_{\sigma^a,\cN^{a,b}} \circ \cN^{a,b})(\rho^a) \bigr) \Big) \, .
\end{align}
This concludes the proof of Theorem~\ref{thm_explicit}.



\section{Universal and explicit refinements of other entropy inequalities}
It is well known (see e.g.~\cite{ruskai02}) that the monotonicity of the relative entropy under trace-preserving completely positive maps is closely related to (i) strong subadditivity, (ii) concavity of the conditional entropy, and (iii) joint convexity of the relative entropy. Based on this relation, Remark~\ref{rmk_mainSimple} can be used to derive universal and explicit remainder terms for the statements (i)-(iii). We note that universal means that the recovery map that occurs in the remainder term does not depend on all possible parameters. Within this section we assume that $A$, $B$, and $C$ are finite-dimensional Hilbert spaces. For $n \in \N$, the $n$-simplex is denoted by $\Delta_{n}:=\{  x\in\R^{n} : \ x\geq 0, \ \sum_{i=1}^{n} x_{i}=1\}$. 

Applying Remark~\ref{rmk_mainSimple} for $\sigma=\id_A \otimes \rho_{BC} \in \Pos(A\otimes B \otimes C)$, $\rho=\rho_{ABC} \in \D_{\sigma}(A\otimes B \otimes C)$, and $\cN_{ABC \to AB}=\cI_{AB}\otimes\tr_C$ implies a strengthened version of the result that has been established in~\cite[Theorem~2.1, Corollary~2.4, and Remark~2.5]{SFR15}.
\begin{corollary}[Strong subadditivity] \label{cor_SSA}
For any density operator $\rho_{ABC} \in \D(A\otimes B \otimes C)$ we have
  \begin{align}
  I(A:C|B)_{\rho} \geq -2 \log F\bigl(\rho_{A B C}, \cR_{B \to BC}(\rho_{A B})\bigr) \ ,
  \end{align}
where $\cR_{B \to BC}(\cdot)= \int_{\R} \di t \beta_0(t) \cR^t_{\rho_{BC},\tr_C}(\cdot)$ as defined in Theorem~\ref{thm_explicit}.
\end{corollary}


 From Corollary~\ref{cor_SSA} we can deduce a universal remainder term for the concavity of the conditional entropy. We note that a similar remainder term has been conjectured in~\cite{BLW14}. Furthermore, a remainder term that however is neither universal nor explicit has been proven in~\cite[Corollary~12]{Wilde15}.

\begin{corollary}[Concavity of conditional entropy]\label{cor_concCondEnt} \
For any $\rho_{AB} \in \D(A\otimes B)$ we have
\begin{align}
H(A|B)_{\rho} - \sum_{x \in \cX} \nu(x) H(A|B)_{\rho^x}  
\geq - 2 \log \sum_{x\in \cX} \nu(x)  F\bigl(\rho_{AB}^x,\cR_{B \to AB}(\rho_B^x)\bigr)  \ ,
\end{align}
for all ensembles $\{\nu(x),\rho_{AB}^x \}_{x\in \cX}$ with $\nu \in \Delta_{\cX}$ and $\rho_{AB}^x \in \D(A \otimes B)$ such that $\rho_{AB}=\sum_{x \in \cX}\nu(x) \rho_{AB}^x$ where $\cR_{B \to AB}(\cdot)=\int_{\R} \di t \beta_0(t) \cR^{t}_{\rho_{AB},\tr_A}(\cdot)$ as defined in Theorem~\ref{thm_explicit}.
\end{corollary}
\begin{proof}
Let us consider the following classical-quantum state
\begin{align}
\phi_{XAB} = \sum_{x \in \cX} \nu(x) \proj{x}_X \otimes \rho_{AB}^x \ .
\end{align}
By Corollary~\ref{cor_SSA} and a simple fact ensuring that the fidelity for orthogonal input states decomposes (see Lemma~A.1 in~\cite{SFR15}), we find
\begin{align}
H(A|B)_{\rho} - \sum_{x \in \cX} \nu(x) H(A|B)_{\rho^x}  
&= H(A|B)_{\phi} - H(A|BX)_{\phi} \\
&= I(X:A|B)_{\phi}  \\
& \geq -2 \log F\bigl(\phi_{XAB} , \cR_{B \to AB}(\phi_{XB}) \bigr) \\
&= -2 \log \sum_{x \in \cX} \nu(x) F\bigl(\rho_{AB}^x , \cR_{B \to AB}(\rho_B^x) \bigr) \ .
\end{align}
We note that by Corollary~\ref{cor_SSA} the recovery map $\cR_{B \to AB}$ is explicit and universal in the sense that it only depends on $\rho_{AB}$. 
\end{proof}

Remark~\ref{rmk_mainSimple} implies a universal remainder term for the joint convexity of the relative entropy. We note that a similar remainder term has been conjectured in~\cite{SBW14}. Corollary~13 in~\cite{Wilde15} proves such a remainder term that however is neither universal nor explicit.
\begin{corollary}[Joint convexity of relative entropy] \label{cor_jointConvRelEnt} \
For any $\nu \in \Delta_{\cX}$, $\{ \sigma^x \}_{x \in \cX}$ with $\sigma^x \in \Pos(A)$ for all $x \in \cX$, $\sigma_{XA}=\sum_{x \in \cX} \nu(x) {\proj{x}_X \otimes \sigma^x}$, any $\rho_x \in \D_{\sigma^x}(A)$ and $\rho=\sum_{x \in \cX} \nu(x) \rho^x$ we have
\begin{align}
\sum_{x \in \cX} \nu(x) D(\rho^x \dd \sigma^x) - D(\rho \dd \sigma) 
\geq -2 \log \sum_{x\in \cS} \nu(x) F\bigl(\rho^x, \cR_{\sigma_{XA},\tr_X}(\rho) \bigr)  \ ,
\end{align}
where $\cR_{\sigma_{XA},\tr_X}(\cdot) = \int_{\R} \di t \beta_0(t) \cR^t_{\sigma_{XA},\tr_X}(\cdot)$ as defined in Theorem~\ref{thm_explicit}.
\end{corollary}
\begin{proof}
Let us consider the states
\begin{align}
 \rho_{XA}:= \sum_{x \in \cX} \nu(x) \proj{x}_X \otimes \rho^x \quad \textnormal{and} \quad \rho = \rho_A = \sum_{x \in \cX} \nu(x) \rho^x  \ .
\end{align}
Since the relative entropy decomposes for orthogonal input states (see Lemma~B.2 in~\cite{SFR15}), Remark~\ref{rmk_mainSimple} implies that the recovery map $\cR_{\sigma_{XA},\tr_X}(\cdot) = \int_{\R} \di t \beta_0(t) \cR^t_{\sigma_{XA},\tr_X}(\cdot)$ satisfies
\begin{align}
\sum_{x \in \cX} \nu(x) D(\rho^x \dd \sigma^x) - D(\rho \dd \sigma)  
&= D(\rho_{XA} \dd \sigma_{XA}) - D(\rho_A \dd \sigma_A) \\
 &\geq -2 \log F\bigl(\rho_{XA} , \cR_{\sigma_{XA},\tr_X}(\rho) \bigr) \\
 &=-2 \log \sum_{x \in \cX} \nu(x) F\bigl(\rho^x , \cR_{\sigma_{XA},\tr_X}(\rho) \bigr) \ ,
\end{align}
where the final step uses that the fidelity for orthogonal input states decomposes (see Lemma~A.1 in~\cite{SFR15}).
\end{proof}


\section{Approximate quantum error correction}
Theorem~\ref{thm_explicit} allows for establishing an alternate, information-theoretic
set of conditions for approximate quantum error correction. Recall that in
quantum error correction the main goal is to protect a set of states from the
action of a noisy quantum channel, by encoding this set into a subspace of a
given Hilbert space. The quantum error correction conditions given in the seminal works \cite{BDSW96,KL97} are necessary and
sufficient for perfect quantum error correction. Although
these works were essential for gaining an understanding of quantum error
correction, the conditions given represent a strong idealization, since we can never hope in
practice to have perfect quantum error correction. So this realization
motivated some follow-up works on approximate quantum error correction (see
\cite{BK02,Mandayam2012} and references therein), in which one seeks to
determine conditions under which approximate recovery from the action of a
given noisy channel is possible. We mention that some of these works made
essential use of the Petz recovery map.

In \cite{HMPB11}, information-theoretic conditions for perfect error
correction were given in terms of the quantum relative entropy. We recall
their result:\ Let $A$ be a Hilbert space, let $C$ be a subspace of $A$
(called the \textquotedblleft codespace\textquotedblright), and let $\Pi$
denote the projection onto $C$. Then%
\begin{equation}
\{\forall\, \rho\in \D(C)\ \textnormal{ : } \ D(\rho\Vert\Pi)=D(\mathcal{N}(\rho)\Vert\mathcal{N}%
(\Pi)) \}
\iff \{\forall \, \rho\in \D(C) \ \textnormal{ : } \ \rho=(\mathcal{P}_{\Pi,\mathcal{N}%
}\circ\mathcal{N})(\rho) \} \ ,\label{eq:rel-ent-cond}%
\end{equation}
where $\cP_{\Pi,\cN}(\cdot)$ is the Petz recovery map defined in~\eqref{eq_Petzmap}.
\textit{So perfect error correction is possible if and only if for all }%
$\rho\in \D(C)$ \textit{the pairwise distinguishability of }$\rho$\textit{ with
}$\Pi$\textit{ does not decrease under the action of the noisy channel
}$\mathcal{N}$.

Given the statement in \eqref{eq:rel-ent-cond} and the prior work on
approximate quantum error correction, it is natural to wonder whether a robust
version of \eqref{eq:rel-ent-cond} exists. Indeed, Theorem~\ref{thm_explicit} (or rather Remark~\ref{rmk_mainSimple}) implies
such a statement, establishing necessary and sufficient information-theoretic
conditions for approximate quantum error correction. Loosely speaking, we can
now say that \textit{approximate error correction is possible if any only if
the pairwise distinguishability of }$\rho$\textit{ with }$\Pi$\textit{ does
not decrease by too much under the action of the noisy channel }$\mathcal{N}$ \emph{for all $\rho \in \D(C)$}.
The corollary below is a consequence of Remark~\ref{rmk_mainSimple} and some other
known facts.

\begin{corollary} \label{cor_QEC}
Let $\varepsilon\in\left[  0,1\right]$, $A$ and $B$ be a finite-dimensional Hilbert spaces, $C$ be a subspace of $A$ (called the
\textquotedblleft codespace\textquotedblright), $\cN \in \TPCP(A,B)$ a quantum channel, and let $\Pi$ denote the
projection onto $C$. If for all $\rho\in \D(C)$ we have 
\begin{equation}
 D(\rho\Vert\Pi)-D\bigl(\mathcal{N}(\rho)\Vert\mathcal{N}(\Pi)\bigr)\leq\varepsilon \ ,
\end{equation}
then it is possible to approximately recover every state $\rho \in \D(C)$, in the sense that we have for all $\rho\in \D(C)$
\begin{equation}
 F\bigl(\rho,(\mathcal{R}_{\Pi,\mathcal{N}}\circ\mathcal{N})(\rho)\bigr)\geq1-\frac{1}{2} \varepsilon\ .\label{eq:fid-approx-QEC}%
\end{equation}
Conversely, if \eqref{eq:fid-approx-QEC} holds for $\varepsilon\in\left[  0,1 \right]  $, then we have for all $\rho\in \D(C)$
\begin{equation}
 D(\rho\Vert\Pi)-D(\mathcal{N}(\rho)\Vert\mathcal{N}(\Pi
))\leq\sqrt{\varepsilon}\log(\dim C)+h_{2}(\sqrt{\varepsilon}) \ ,
\end{equation}
where $h_{2}(p):=-p\log p-(1-p)\log(1-p)$ is the binary entropy, with the
property that $\lim_{p\searrow0}h_{2}(p)=0$.
\end{corollary}

\begin{proof}
The first statement is a direct consequence of Remark~\ref{rmk_mainSimple}, found by
setting $\rho=\rho$, $\sigma=\Pi$, and $\mathcal{N}=\mathcal{N}$ and then applying the inequality $-\log(x)\geq1-x$,
which holds for $x\in\left[  0,1\right]  $.

To prove the second statement, suppose that \eqref{eq:fid-approx-QEC} holds.
By the Fuchs-van-de-Graaf inequality~\cite{fuchs99}, \eqref{eq:fid-approx-QEC} implies that for all $\rho \in \D(C)$ we have
\begin{equation}
\frac{1}{2}\left\Vert \rho-(\mathcal{R}_{\Pi,\mathcal{N}%
}\circ\mathcal{N})(\rho)\right\Vert _{1}\leq\sqrt{\varepsilon} \ .
\end{equation}
Then consider the following chain of inequalities, which holds for all $\rho\in \D(C)$:
\begin{align}
& \!\!\!\! D(\rho\Vert\Pi)-D(\mathcal{N}(\rho)\Vert\mathcal{N}(\Pi)) \notag \\
& \leq D(\rho\Vert\Pi)-D\bigl((\mathcal{R}_{\Pi,\mathcal{N}}\circ\mathcal{N})(\rho)\Vert(\mathcal{R}_{\Pi,\mathcal{N}}\circ\mathcal{N})(\Pi)\bigr)  \\
& =D(\rho\Vert\Pi)-D\bigl((\mathcal{R}_{\Pi,\mathcal{N}}\circ\mathcal{N})(\rho)\Vert\Pi\bigr) \\
& =-H(\rho)+H((\mathcal{R}_{\Pi,\mathcal{N}}\circ\mathcal{N})(\rho
))+\tr\bigl(\left[  (\mathcal{R}_{\Pi,\mathcal{N}}\circ
\mathcal{N})(\rho)-\rho\right]  \log\Pi\bigr) \\
& \leq\sqrt{\varepsilon}\log\dim(C)+h_{2}(\sqrt{\varepsilon}).
\end{align}
The first inequality is a consequence of monotonicity of quantum relative
entropy with respect to the recovery map $\mathcal{R}_{\Pi,\mathcal{N}}$. The
first equality follows because $(\mathcal{R}_{\Pi,\mathcal{N}}\circ
\mathcal{N})(\Pi))=\Pi$. The second equality is a rewriting, and the last
inequality follows from the Fannes-Audenaert inequality~\cite{audenaert07} (see also~\cite[Lemma~1]{winter15}) and the facts that $\operatorname{supp}(\rho
)\subseteq\operatorname{supp}(\Pi)$, $\operatorname{supp}((\mathcal{R}%
_{\Pi,\mathcal{N}}\circ\mathcal{N})(\rho))\subseteq\operatorname{supp}(\Pi)$,
and $\log\Pi=0$ on the support of $\Pi$.
\end{proof}

\section{Conclusion}
In this work, we showed that for any non-negative operator $\sigma$ and for any channel $\cN$ there exists an \emph{explicit} and \emph{universal} recovery map $\cR_{\sigma,\cN}$ such that
\begin{align} \label{eq_res_conc}
D(\rho \| \sigma) - D\bigl(\cN(\rho) \| \cN(\sigma) \bigr) \geq - 2 \log F\bigl(\rho,(\cR_{\sigma,\cN}\circ\cN)(\rho)\bigr)
\end{align}
for all density operators $\rho$ such that $\supp(\rho) \subseteq \supp(\sigma)$. A consequence of its universality is that the recovery map $\cR_{\sigma,\cN}$ satisfies $(\cR_{\sigma,\cN}\circ\cN)(\sigma)=\sigma$. The present work thus constitutes an improvement over~\cite{Wilde15} since the recovery map is constructed explicitly and is independent of $\rho$.
Knowing the explicit structure of the recovery map that satisfies \eqref{eq_res_conc} can be helpful in various scenarios. For example if this is to be implemented by experimentalists or in an actual quantum computer it is helpful to have an explicit recovery map. 
 We showed that the inequality~\eqref{eq_res_conc} implies universal and explicit remainder terms for other entropy inequalities such as strong subadditivity, concavity of the entropy, and the joint convexity of the relative entropy. Furthermore, it can be useful in the context of quantum sufficient statistics and approximate quantum error correction. 



\section*{Appendix}

\appendix

\section{Proof of Lemma \ref{thm:op-hirschman}}

\label{sec:op-hirschman-proof}

We begin by recalling Hirschman's strengthening \cite{H52} of the Hadamard
three-line theorem (see, e.g., \cite[Lemma~1.3.8]{G08}):

\begin{lemma}
\label{thm:hirschman}Let $S:=\left\{  z\in\mathbb{C}:0\leq\operatorname{Re}\left\{  z\right\} \leq1\right\}$ and let $g(z)$ be holomorphic on the interior of $S$ and continuous on the
boundary, such that%
\begin{equation}
\sup_{z\in S}\exp\left\{  -a\left\vert \operatorname{Im}z\right\vert \right\}
\log\left\vert g(z)\right\vert \leq A<\infty,
\end{equation}
for some fixed $A$ and $a<\pi$. Then for $\theta\in(0,1)$, the following bound
holds%
\begin{equation}
\log\left\vert g(\theta)\right\vert \leq\int_{\R}\di t\ \Bigl(
\alpha_{\theta}(t)\log\bigl(  \left\vert g(\ci t)\right\vert ^{1-\theta}\bigr)
+\beta_{\theta}(t)\log\bigl(  \left\vert g(1+\ci t)\right\vert ^{\theta}\bigr)
\Bigr)  ,
\end{equation}
where $\alpha_{\theta}(t)$ and $\beta_{\theta}(t)$ are defined in \eqref{eq:pt}.
\end{lemma}

\noindent We can now prove Lemma~\ref{thm:op-hirschman}.

\begin{proof}[Proof of Lemma~\ref{thm:op-hirschman}]
The proof of this theorem is well known, but we provide it for completeness. For fixed $\theta\in(0,1)$, let $q_{\theta}$ be the H\"{o}lder
conjugate of $p_{\theta}$, defined by%
\begin{equation}
\frac{1}{p_{\theta}}+\frac{1}{q_{\theta}}=1 \ .
\end{equation}
Similarly, let $q_{0}$ and $q_{1}$ be H\"{o}lder conjugates of $p_{0}$ and
$p_{1}$, respectively. We can find an operator $X$ such that%
\begin{align}
\left\Vert X\right\Vert _{q_{\theta}}   =1 \quad \textnormal{and} \quad
\operatorname{tr}(  XG(\theta))     =\left\Vert G(\theta)\right\Vert
_{p_{\theta}} \ .
\end{align}
We can write the singular value decomposition for $X$ in the form
$X=UD^{1/q_{\theta}}V$ (implying tr$(D)=1$). For $z\in S$, define%
\begin{equation}
X(z):= UD^{\frac{1-z}{q_{0}}+\frac{z}{q_{1}}}V \ .
\end{equation}
As a consequence, $X(z)$ is holomorphic on the interior of $S$ and continuous on the boundary. Also, observe that $X(\theta)=X$. Then the
following function satisfies the requirements needed to apply
Lemma~\ref{thm:hirschman}:%
\begin{equation}
g(z):=\operatorname{tr}\bigl(  X(z)G(z) \bigr) \   .
\end{equation}
Indeed, we have that%
\begin{equation}
\log\  \left\Vert G(\theta)\right\Vert _{p_{\theta}}
=\log\left\vert g(\theta)\right\vert \leq\int_{\R}\di t\ \left(
\alpha_{\theta}(t)\log\bigl(  \left\vert g(\ci t)\right\vert ^{1-\theta}\bigr)
+\beta_{\theta}(t)\log\bigl(  \left\vert g(1+\ci t)\right\vert ^{\theta}\bigr)
\right)  .\label{eq:app-hirschman}%
\end{equation}
Now, from applying H\"{o}lder's inequality and the facts that $\left\Vert
X(it)\right\Vert _{q_{0}}=1=\left\Vert X(1+\ci t)\right\Vert _{q_{1}}$, we find
that%
\begin{align}
\left\vert g(\ci t)\right\vert  & =\left\vert \operatorname{tr}\bigl(
X(\ci t)G(\ci t)\bigr)  \right\vert \leq\left\Vert X(\ci t)\right\Vert _{q_{0}%
}\left\Vert G(\ci t)\right\Vert _{p_{0}}=\left\Vert G(\ci t)\right\Vert _{p_{0}}
\end{align}
and
\begin{align}
\left\vert g(1+\ci t)\right\vert  & =\left\vert \operatorname{tr}\bigr(
X(1+\ci t)G(1+\ci t)\bigr)  \right\vert \leq\left\Vert X(1+\ci t)\right\Vert _{q_{1}%
}\left\Vert G(1+\ci t)\right\Vert _{p_{1}}=\left\Vert G(1+\ci t)\right\Vert _{p_{1}%
}.
\end{align}
Bounding \eqref{eq:app-hirschman} from above using these inequalities then
gives \eqref{eq:oper-hirschman}.
\end{proof}


\section*{Acknowledgments}

We thank Omar Fawzi, Rupert Frank, J\"urg Fr\"ohlich, Elliott Lieb, Volkher Scholz and Marco Tomamichel for helpful discussions. We further thank the anonymous referee for constructive feedback.
MJ's work was supported by NSF DMS 1501103 and NSF-DMS 1201886. 
RR and DS acknowledge support by the European Research Council (ERC) via grant No.~258932, by the Swiss National Science Foundation (SNSF) via the National Centre of Competence in Research ``QSIT'', and by the European Commission via the project ``RAQUEL''. 
MMW acknowledges support from the NSF under Award Nos.~CCF-1350397 \& 1714215 and the DARPA Quiness Program through US Army Research Office award W31P4Q-12-1-0019. He is also grateful to co-authors MJ, RR, and AW for hospitality and support during research visits in summer 2015.
AWÕs work was supported by the EU (STREP ÒRAQUELÓ), the ERC (AdG ÒIRQUATÓ), the Spanish MINECO (grant FIS2013-40627-P) with the support of FEDER funds, as well as by the Generalitat de Catalunya CIRIT, project 2014-SGR-966.

\vspace{-3mm}


  \bibliographystyle{abbrv}
  
  \bibliography{bibliofile}

\end{document}